\documentclass[runningheads]{llncs}
\usepackage[T1]{fontenc}

\usepackage{graphicx}
\usepackage{auxlib}
\usepackage{macro}

\usepackage{microtype}
\usepackage{comment}
\usepackage{amsthm}
\usepackage{amsmath,amssymb}
\usepackage{algorithm}
\usepackage{algorithmicx}
\usepackage{algpseudocode}

\floatname{algorithm}{Algorithm} 

\begin{document}

\title{Optimal Repurchasing  Contract Design for Efficient Utilization of Computing Resources}

\author{Zhengyan Deng\inst{1}\and
Yusen Zheng\inst{2}\and
Chenliang Sheng\inst{3}\and Shaowen Qin\inst{4}
}
\authorrunning{Z. Deng et al.}

\titlerunning{Contract Design for Computing Resources Repurchasing}
\institute{Jiangnan University, Wuxi, China
\and
Peking University, Beijing, China
\and
Hefei University of Technology, Hefei, China
\and
Flinders University, 
Adelaide, Australia\\
\email{6240910002@stu.jiangnan.edu.cn, yusen@stu.pku.edu.cn, 2023212716@mail.hfut.edu.cn, shaowen.qin@flinders.edu.au}}
\maketitle              

\begin{abstract}

The rapid advancement of AI and other emerging technologies has triggered exponential growth in computing resources demand. 
Faced with prohibitive infrastructure costs for large-scale computing clusters, users are increasingly resorting to leased computing resources from third-party providers.
However, prevalent overestimation of operational requirements frequently leads to substantial underutilization of the computing resources.
To mitigate such inefficiency, we propose a contract-based incentive framework for computing resources repurchasing. Comparing to auction mechanisms, our design enables providers to reclaim and reallocate surplus computing resources through market-driven incentives.
Our framework operates in a multi-parameter environment where both clients’ idle resource capacities and their unit valuations of retained resources are private information, posing a significant challenge to contract design. 
Two scenarios are considered based on whether all clients possess the same amount of idle resource capacity. 
By transforming the contract design problem into solving a mathematical program, we obtain the optimal contracts for 
each scenario, which can maximize the utility of computing resources providers while ensuring the requirements of incentive compatibility (IC) and individual rationality (IR).  
This innovative design not only provides an effective approach to reduce the inefficient utilization of computing resources, but also establishes a market-oriented paradigm for sustainable computing ecosystems.

\end{abstract}
\setcounter{footnote}{0}
\newpage
\section{Introduction}
The rapid development of the digital economy and advancements in artificial intelligence (AI) have positioned computing resources as a key driver of modern productivity. 
According to the ``2023–2024 China Artificial Intelligence Computing Resources Development Assessment Report'', published by IDC, the scale of China's intelligent computing resource reached 260 EFLOPS in 2022 and is projected to exceed 1,117 EFLOPS by 2027, with a remarkable compound annual growth rate (CAGR) of 33.9\%
\footnote{\url{https://www.ieisystem.com/global/file/2023-12-01/17014097286402c975afc8bfb91fe59018c23ec288049fd.pdf}}. 
This rapid growth has been accompanied by a sharp increase in the demand for high-performance GPUs. 
The imbalance between supply and demand has led to rising prices and frequent shortages, making it increasingly challenging to acquire high-performance GPUs. 
In response, an increasing number of clients are turning to the computing resource leasing market.

Traditional centralized cloud computing infrastructures are increasingly unable to meet these growing computational requirements. Computing resources leasing is an emerging service model designed to provide clients with flexible and efficient computing  resources.\footnote{\url{https://www.21jingji.com/article/20231212/herald/4a5f93fbee91a636d7f324ea9ea69efd.html}} It is also better equipped to meet the increasingly diverse demands of the current internet industry. Under the computing resources leasing model, clients can lease computing resources from third-party providers based on their specific needs to carry out computational tasks, without establishing extensive computing infrastructure locally.

However, existing computing resources leasing models, exhibit a common issue of low computing resources utilization efficiency. The issue has become more critical in the current era of massive demand for computing resources.
According to IDC data, the utilization rate of general-purpose computing centers in China, which primarily serve enterprises, is only 10\% to 15\%\footnote{\url{https://news.qq.com/rain/a/20241029A06DME00}}. This indicates that a significant portion of computing  resources remains idle due to a mismatch between supply and actual demand.
Meanwhile, some existing clients—though actively consuming computing resources—allocate more resources than necessary to non-essential tasks, obscuring the true demand.
It is a challenge for computing resources providers to accurately detect the actual effective utilization of their resources.
Bridging these gaps through traditional infrastructure expansion often requires significant capital expenditures, extended implementation timelines, and increased energy demands, resulting in systemic inefficiencies that ultimately leads to substantial waste of resources on a social economical scale.

One economical and effective approach to address this problem is to reclaim idle resources of current clients through economic incentives. To this end, we focus on designing contract mechanisms that encourage clients to sell back their idle computing resources to the provider, who can then reallocate these recovered resources to other clients in need.
The fundamental idea behind our contract design is to enable clients with idle resources to maximize their payoffs by carefully selecting a contract item that aligns with their true type. Specifically, we define a client’s type as a tuple  $(v,c)$ that includes their capacity $c$ of idle computing resources and their valuation $v$ per unit of resource.
On the other hand, the provider aims to maximize his utility, which is the difference between the revenue generated by reallocating the reclaimed resources to new clients and the cost incurred in repurchasing these resources. To achieve this goal, the provider designs customized contract items tailored to different client types and presents them as a contract menu. Clients then select a contract item from the menu (or choose not to participate) based on their type. Once a contract is chosen, clients return the resources as stipulated in the agreement and receive the corresponding compensation.

The main contribution of this paper is the design of an implementable contract that maximizes the provider’s utility in a multi-dimensional private information environment (private capacity of idle resources and private valuation for one unit of resource) while simultaneously satisfying the requirements of individual rationality (IR) and incentive compatibility (IC) constraints.
To address the challenges brought by the two-parameter setting, we assume that the provider knows the discrete probability distribution of clients' types $(v,c)$. 
Such a setting can effectively help us reduce the computational complexity of the problem.
Based on this assumption, we innovatively transform the contract design problem into a mathematical program. To be specific, the objective function of this program is formulated to maximize the expected utility of the provider, while the constraints are carefully designed to ensure IC, IR, and other feasibility requirements. Through this rigorous formulation, we are able to derive optimal contracts for the computing resources repurchasing problem.

\section{Related Works}

To meet the high demand of computing resources associated with the fast development of AI, large corporations build data centers, while small corporations that cannot afford such infrastructure choose to rely on leasing computing resources to address resource demand in AI research. \cite{sun2024}. Research on computing resources leasing remains limited, and even fewer studies focus on reclaiming idle resources. However, computing resources leasing is fundamentally very much similar to other leasing services, and reclaiming idle resources is comparable to utilizing idle virtual resources like CPU and bandwidth.

\textit{Mechanism Design for Idle Digital Resource Reutilization.} Numerous studies have been conducted on the recycling or utilization of idle computer resources.
In 2017, Quttoum et al., \cite{quttoumAFAMFairAllocation2017} proposed a fair resource allocation model (AFAM) aimed at optimizing resource allocation in Cloud Data Center Networks (CDNs) and improving resource utilization through the VCG auction mechanism. Later, in 2024 \cite{quttoumAMADAdaptiveMapping2024}, they introduced a resource reclamation model (AMAD) based on the Stackelberg leadership framework to reclaim idle resources from users who have already leased resources by employing repeated leader-follower games, thereby further enhancing the actual resource utilization rate.
Hosseini et al., \cite{hosseiniCrowdcloudCrowdsourcedSystem2019} proposed a crowdsourced cloud infrastructure called Crowdcloud, which aims to repurpose users' idle computing resources through crowdsourcing methods, thereby creating a decentralized cloud computing platform.
Muktadir et al. \cite{henaalmuktadirResourceNegotiationGame2018} proposed a leader-follower game-based mechanism for reclaiming unused virtual resources. The infrastructure provider acts as the leader and virtual network operators as followers, negotiating iteratively to set the compensation price for reclaimed resources. And in 2019 \cite{muktadirRepeatedLeaderFollower2019}, the authors refined their original model by further categorizing the types of VNOs and introducing a strategy that adjusts resource compensation prices based on historical negotiation data. 
Liu et al., \cite{liuGameTheorybasedOptimization2020} approached the problem of reclaiming idle computing resources from organizations by modeling it as a non-cooperative game. The study proposed a dynamic pricing mechanism to integrate and redistribute these idle computing resources, making them available for other cloud users.

\textit{Contract-based Mechanism Design.} 
Contract-based incentive mechanism has been widely used in such fields as resource sharing and crowdsensing areas.
Zhang et al., \cite{zhangContractBasedIncentiveMechanisms2015} developed a contract-theoretic incentive mechanism to encourage user participation in Device-to-Device communications within cellular networks. This mechanism models the interaction between the base station and users, offering performance-reward contracts that effectively boost user engagement and significantly increase network capacity compared to other methods.
Ma et al., \cite{maIncentivizingWiFiNetwork2018} proposed a contract-theory-based incentive mechanism to promote Wi-Fi resource sharing in crowdsourced wireless community networks. This mechanism aims to encourage users to share their private Wi-Fi access points, thereby expanding Wi-Fi coverage and reducing deployment costs for operators.
Dai et al., \cite{daiTrustDrivenContractIncentive2022} proposed a trust-driven contract incentive framework to address the issues of trust and incentives in mobile crowdsensing networks. This contract incentive scheme takes into account users’ privacy preferences, to design a set of optimal contracts that maxi- mize the utility of both users and platforms.
Zhao et al., \cite{zhaoContractBasedIncentiveMechanism2024} considered mobile users' varying privacy preferences and potential information asymmetry from users' passive disclosure of preferences, the study proposed a contract theory-driven incentive mechanism to optimize the balance between resource consumption and task completion under conditions of information asymmetry.
Xie et al., \cite{xieIncentiveMechanismDesign2025} proposed an incentive mechanism design based on contract theory for resource trading in computational power networks. This mechanism takes into account the trust between resource providers and consumers, and promotes the effective transaction of resources by designing a credible incentive mechanism.

Although numerous studies have focused on the utilization of idle virtual resources, research on reclaiming leased computing resources remains scarce. In particular, there has been little discussion of computing resource reclamation from the perspective of contract theory. In practical computing resources leasing scenarios, clients not only conceal their true valuation of computing resources but also hide the actual scale of resources they are using, which increases the complexity of the problem beyond that of traditional studies. Given this,  it is an innovative attempt to introduce contract theory to systematically explore the issue of computing resource reclamation.

\section{Preliminaries}\label{sec:pre}

This section firstly introduces the problem of repurchasing computing resources and then formally define the contract, along with the desired properties of the contract design.

\subsection{The Computing Resources Repurchasing Problem}
\label{sec:formalization}

In the computing resources repurchasing problem, a \emph{provider} $P$ offers computing resources, and there are $n$ \emph{clients} who have already rented these resources. 
Each current client may have idle resources.
When new clients arrive seeking resources, the provider may not have enough available. 
Therefore, the provider $P$ needs to recover idle resources from current clients and release them to the new clients.
To ensure the minimum demand of the new coming clients is met, a lower bound on supply available for release (denoted by $D$) is set in advance. 
If the provider $P$ fails to collect enough resources to meet $D$, he will face a penalty.
This penalty is justified, as failing to satisfy the demand of new clients would force the provider to either lease computing resources from other providers to retain these clients or lose them, both of which would incur a loss.
Therefore, provider $P$ aims to maximize his utility by designing \emph{contract} that incentivize clients to return their idle resources, which will be explained in detail in \cref{sec:contracts}.

Let $N = \set*{1,2,\cdots,n}$ denote the current client set. Each current client $i$ has a \emph{valuation} $v_i$ for a unit of computing resource and a \emph{capacity} $c_i$ of her idle resources. We assume that the valuations of all clients are from a set, denoted by $V=\{v^1,v^2,\cdots,v^K\}$, which contains $K$ distinct values with $v^1< v^2<\cdots< v^K$. Similarly, the capacities of idle resources are also assumed to be from a set, denoted by $C=\{c^1,c^2,\cdots,c^L\}$, with $c^1< c^2< \cdots<c^L$. 
Therefore, the private information of each client $i$ can be characterized by a tuple $(v_i,c_i)\in V\times C$.
We refer to the private information of client $i$ as the \emph{type} of client $i$.
For simplicity, let $\Gamma = V \times C$ denote the set of all possible types, and let $\gamma_i = (v_i, c_i)$ represent the type of client $i$.
The dual private nature of $v_i$ and $c_i$ induces a multi-parameter mechanism design problem.
We further assume each client’s type $(v_i,c_i)$, $i\in N$, is independently drawn from a finite-support joint discrete distribution $\Delta_i$, which is privately known to the provider.
This is because the provider can infer these distributions from the previous historical transaction data.
Let $\lambda_i^{l,k}$ denote the probability that client $i$'s type is realized as $(v^k, c^l)$.
Therefore, we have $\sum_{k,l}\lambda_i^{l,k}= 1$ for all $i\in N$.

\subsection{Contract}
\label{sec:contracts}

To facilitate the repurchasing of computing resources, we introduce a novel contract-based framework that empowers providers to incentivize clients to contribute their idle resources through tailored economic rewards. 
Informally, a contract specifies a set of contract items, each containing a recommended repurchasing amount and a payment price. 
Clients have the option to sign the contract and select one of items based on their types. Alternatively, they may choose not to sign the contract. In this case, they are not required to return any resources and will not receive any payment.

\begin{definition}[Contract]\label{def:contract}
A \emph{contract}  is defined as a pair of functions $(x, p)$, where $x: \Gamma \to \RR_+$ is an allocation rule that maps each type to a recommended repurchasing amount, and $p: \Gamma \to \RR_+$ is a payment rule that maps each type to a non-negative payment.
\end{definition}

For convenience, we slightly abuse the notation and use $x_{v^k}^{c^l}$ and $p_{v^k}^{c^l}$ to denote the recommended repurchasing amount and payment for type $(v^k, c^l)$, respectively. That is, $x_{v^k}^{c^l} \deq x(v^k, c^l)$ and $p_{v^k}^{c^l} \deq p(v^k, c^l)$.
Since the type set $\Gamma$ is the domain of the allocation and payment rules of the contract, we refer to each type $\gamma \in \Gamma$ as a \emph{contract item} in this context.
In contract design, we only consider contracts where all payments are non-negative, i.e., $p_{v^k}^{c^l} \geq 0$ for all $(v^k, c^l) \in \Gamma$, as these serve as compensation for clients who return resources.

\begin{figure}[H]
  \centering
 \includegraphics[width=0.8\textwidth]{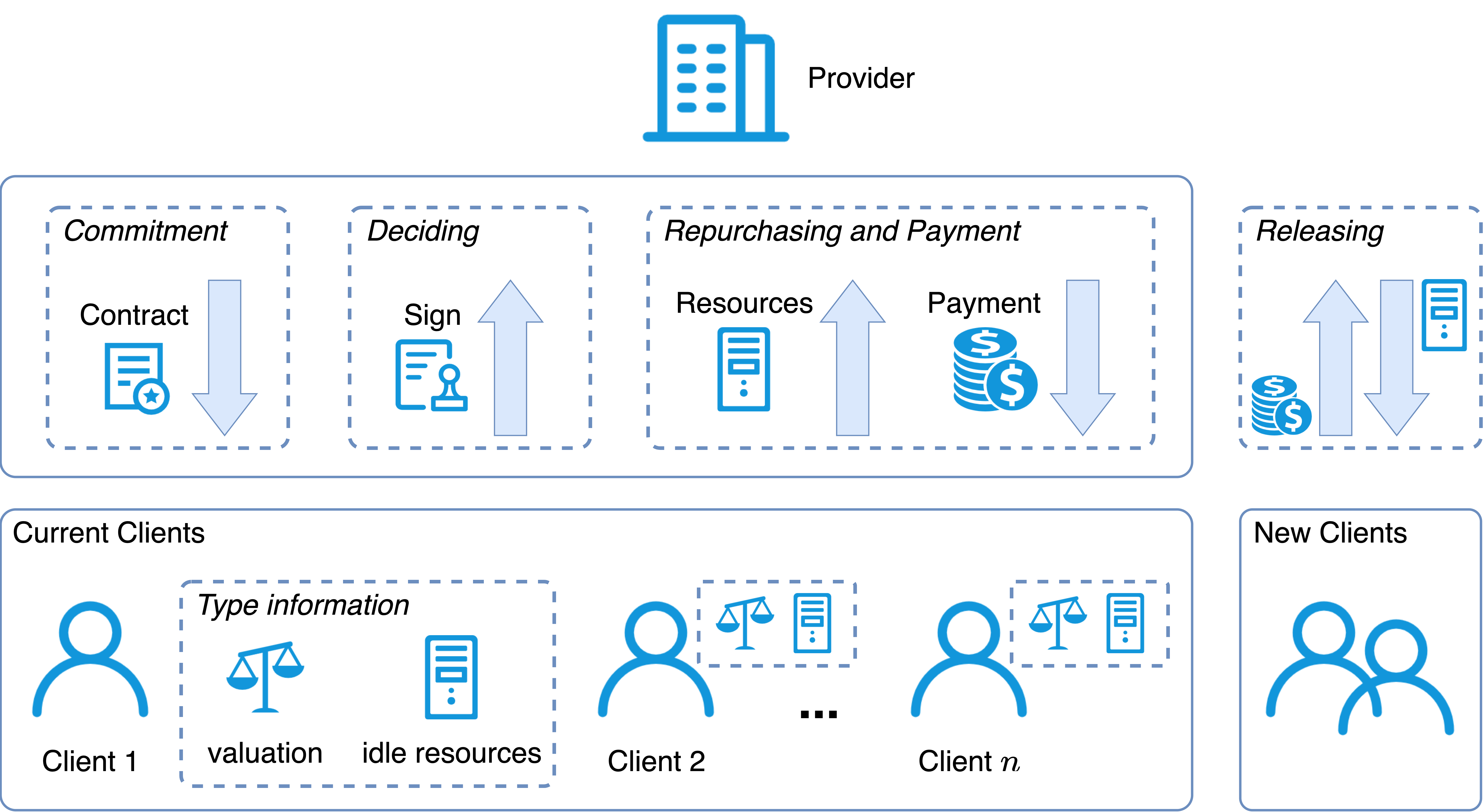}
  \caption{The interactions between the provider and the clients.}
   \label{interactions}  
\end{figure}

The procedure of the interactions between the provider and the clients is as follows (as showns in \cref{interactions}):

 \begin{itemize}
    \item \emph{Commitment.} 
    The provider broadcasts the contract $(x,p)$ to the current clients $\set*{1,2,\cdots,n}$.
    \vspace{0.5em}
    \item \emph{Deciding.} 
    Upon receiving the proposed contract, clients evaluate it based on their true types $\gamma_i = (v_i, c_i)$ and decide whether to sign the contract with the provider. 
    If they choose to sign, they will strategically select one contract item $\gamma_i' = (v_i', c_i')$ to optimize their utility $U_i(\gamma_i'; \gamma_i)$, defined as:
    \begin{eqnarray*}
        U_i(\gamma_i';\gamma_i) = p(\gamma_i') - v_ix(\gamma_i').
    \end{eqnarray*}
    The chosen item $\gamma_i'$ does not need to be the client's true type, but it must satisfy $x(\gamma_i') \leq c_i$, meaning the selected recommended repurchasing amount cannot exceed the client's idle resource capacity.
    Once a contract is signed and an item is chosen, the provider and the client are obligated to strictly adhere to the terms specified in the contract item.
    \vspace{0.5em}
    \item \emph{Repurchasing and Payment.} 
    The client return $x(\gamma_i')$ amount of resource to the provider, and the provider pays $p(\gamma_i')$ to the client.
    \vspace{0.5em}
    \item \emph{Releasing.} After collecting resources from all current clients, the provider shall release them to the new coming clients, and thus obtain corresponding utility, which is then defined as:
    \begin{equation*}
        \begin{aligned}
            &U_P(\gamma_1',\gamma_2',\cdots,\gamma_n')\\
            =&\alpha \cdot \sum_{i \in N} x(\gamma_i') - \sum_{i \in N} p(\gamma_i') + M \cdot \min\left\{0, \sum_{i \in N} x(\gamma_i') - D\right\}
        \end{aligned}
    \end{equation*}
    where $\alpha$ is an exogenous variable, representing the rental price per unit of resource, and $M$ denotes the penalty coefficient applied to each unit of shortfall in resources below the supply lower bound $D$.
\end{itemize}

\subsection{Problem Formulation}
The provider aims to design a contract $(\textbf{x,p})$ that maximizes her expected utility while ensuring several desired properties, including resource feasibility (\cref{def:resource_feasibility}), incentive compatibility (\cref{def:IC1}), and individual rationality (\cref{def:IR}).

The resource feasibility condition ensures that the recommended repurchasing amount does not exceed the client's idle resource capacity when the client selects the contract item that matches her true type.

\begin{definition}[Resource Feasibility]\label{def:resource_feasibility}
A contract is \emph{resource feasible} if, for every contract item \((v^k, c^l)\), the repurchased amount \(x_{v^k}^{c^l}\) does not exceed the available capacity \(c^l\). Specifically, the following condition must hold:
\begin{align}\label{eq:resource_feasibility}
    x_{v^k}^{c^l} \leq c^l, \quad \forall c^l \in C, \quad v^k \in V.
\end{align}
\end{definition}

The resource greedy condition ensures that, for the same private valuation, types with greater capacity are recommended more resources. Additionally, to encourage the return of more resources, the mechanism requires that as many resources as possible be repurchased, for each type $\gamma$.

\begin{definition}[Resource Greedy]\label{def:RG}
    A contract is \emph{resource greedy} if given a valuation $v^k\in V$, the following two conditions hold for all types $(v^k,c^l)$, $\forall c^l\in C$:
    \begin{itemize}
        \item \emph{Monotonicity in Capacity.}
        For any two capacity \(c^l > c^{l'}\), the repurchased amount from the higher-capacity type must be at least as large as that from the lower-capacity type:
        \begin{align}
        x_{v^k}^{c^l} \geq x_{v^k}^{c^{l'}}, \quad\forall c^l>c^{l'}, v^k\in V. \label{rg.1}
        \end{align}
        \item \emph{Maximal Recycling for Dominated Types.}
        If a higher-capacity type \(c^l\) is repurchased strictly more than a lower-capacity type \(c^{l'}\), i.e., \(x_{v^k}^{c^l} > x_{v^k}^{c^{l'}}\), then the repurchased amount from the lower-capacity type must equal its full capacity \(c^{l'}\):
        \begin{align}
        x_{v^k}^{c^{l'}} = c^{l'}, \quad\forall c^l,c^{l'}\in C, v^k\in V, \textnormal{ s.t. } x_{v^k}^{c^l}>x_{v^k}^{c^{l'}}. \label{rg.2}
        \end{align}
    \end{itemize}
\end{definition}

The incentive compatibility condition ensures that clients prefer the contract items specifically designed for their true types.

\begin{definition}[Incentive Compatibility]\label{def:IC1}
A contract is \emph{incentive compatible} if it is resource feasible, and clients achieve the maximum utility by selecting the contract item that matches their true type, i.e.,
\begin{align}
    p_{v^k}^{c^l} - v^k x_{v^k}^{c^l} &\geq p_{v^{k'}}^{c^{l'}} - v^k x_{v^{k'}}^{c^{l'}}, & \forall v^k, v^{k'} \in V,c^l,c^{l'}\in C, \textnormal{ s.t. } x_{v^{k'}}^{c^{l'}} \leq c^{l}. \label{6}
\end{align}
\end{definition}

The inequality in \cref{def:IC1} simultaneously accounts for the misreporting of both valuation and capacity.
\cref{lem:IC2} shows that this can be decoupled into two separate inequalities, each focusing on a single parameter. We thus respectively name them as \emph{the incentive compatibility w.r.t the valuation and the capacity}.

\begin{lemma}[Equivalence of Incentive Compatibility]\label{lem:IC2}
A contract is incentive compatible if and only if  it is resource feasible, resource greedy and satisfies the following conditions:
\begin{align}
    p_{v^k}^{c^l} - v^k x_{v^k}^{c^l} &\geq p_{v^{k'}}^{c^l} - v^k x_{v^{k'}}^{c^l}, & \forall v^k, v^{k'} \in V,c^l \in C. \label{4}\\
    p_{v^k}^{c^l} - v^k x_{v^k}^{c^l} &\geq p_{v^k}^{c^{l'}} - v^k x_{v^k}^{c^{l'}}, & \forall v^k \in V,c^l,c^{l'}\in C \textnormal{ s.t. } x_{v^k}^{c^{l'}} \leq c^{l}. \label{5}
\end{align}
\end{lemma}

The full proof for \cref{lem:IC2} is left in \cref{proof:lem:IC2}.

\begin{definition}[Individual Rationality]\label{def:IR}
A contract is \emph{individually rational} if it satisfies the following condition: each client achieves a non-negative utility when she signs the contract corresponding to her true types, that is
\begin{align}
    p_{v^k}^{c^l} - v^k x_{v^k}^{c^l} &\geq 0, & \forall v^k\in V, c^l\in C. \label{7}
\end{align}
\end{definition}

In this paper, our goal is to design a contract which simultaneously satisfies the properties of resource feasibility and resource greediness, as well as incentive compatibility and individual rationality. For the sake of convenience, we refer to  such a contract as a \emph{feasible} contract. Therefore, given a feasible contract $(\textbf{x,p})$, the expected utility of the provider can be expressed as
\begin{align}\label{EUP}
    &\mathbb E_{\forall i, \gamma_i \sim \Delta_i}[U_P(\gamma_1,\gamma_2,\cdots,\gamma_n)] \nonumber\\
    =& \sum_{i=1}^n \sum_{l=1}^L \sum_{k=1}^K \lambda_{i}^{l,k} \left( \alpha x_{v^k}^{c^l} - p_{v^k}^{c^l} \right) 
    + M \cdot \min\left\{0, \left( \sum_{i=1}^n \sum_{l=1}^L \sum_{k=1}^K \lambda_{i}^{l,k} x_{v^k}^{c^l} - D \right) \right\}. 
\end{align}

Hence, the optimal contract design can be formulated as the following optimization problem:
\begin{align}
\max_{(x,p)} \quad & \Ex_{\forall i, \gamma_i \sim \Delta_i}[U_P(\gamma_1,\gamma_2,\cdots,\gamma_n)] & \text{(Expected utility)} \label{program}\\
\text{s.t.} \quad & \eqref{eq:resource_feasibility},\eqref{rg.1},\eqref{rg.2},\eqref{6},\eqref{7}. & \text{(Feasibility constraints)}\nonumber
\end{align}

\section{Optimal Contract Design for Repurchasing Computing Resources Problem}\label{sec:algorithm}

\subsection{Characterization of Feasible Contracts}

In this section, we temporarily ignore the objective function of the optimization problem and characterize the feasible contracts in advance. 
\cref{thm:feasible} provides a comprehensive characterization of a feasible contract.

\begin{theorem}[Feasible Contract]\label{thm:feasible}
A contract is feasible if and only if the following conditions hold:
\begin{itemize}[left=2em]
    \item[(P1)] Resource feasibility: $x_{v^k}^{c^l} \leq c^l$ for all $v^k \in V$ and $c^l \in C$.
    \vspace{0.5em}
    \item[(P2)] The recommended repurchasing amount decreases as the valuation increases, i.e., $c^l \geq x_{v^p}^{c^l} \geq x_{v^q}^{c^l} \geq 0$ for all $v^q \geq v^p$ and $c^l \in C$.
    \vspace{0.5em}
    \item[(P3)] The payment satisfies the squeeze inequality: $v^p (x_{v^p}^{c^l}-x_{v^q}^{c^l}) \leq p_{v^p}^{c^l} - p_{v^q}^{c^l} \leq v^{q} (x_{v^p}^{c^l}-x_{v^q}^{c^l})$ for all $v^p, v^q \in V$ and $c^l \in C$.
    \vspace{0.5em}
    \item[(P4)] Incentive compatibility w.r.t. the capacity: $p_{v^k}^{c^l} - v^k x_{v^k}^{c^l} \geq p_{v^k}^{c^{l'}} - v^k x_{v^k}^{c^{l'}}$ for all $v^k \in V,c^l,c^{l'}\in C$ subject to $x_{v^k}^{c^{l'}} \leq c^{l}$.
    \vspace{0.5em}
    \item[(P5)] The utility of the client having the highest valuation is non-negative: $p_{v^K}^{c^l} - v^K x_{v^K}^{c^l} \geq 0$ for all $c^l \in C$.
    \vspace{0.5em}
    \item[(P6)] The recommended repurchasing amount increases as the capacity increases (i.e., $\forall c^l>c^{l'}$, $x_{v^k}^{c^l}\geq x_{v^k}^{c^{l'}}$), and if $x_{v^k}^{c^l}> x_{v^k}^{c^{l'}}$, $ x_{v^k}^{c^{l'}}=c^{l'}$.
\end{itemize}
\end{theorem}

Before proving \cref{thm:feasible}, the following two lemmas are necessary.

\begin{lemma}\label{lem:ic}
A contract is incentive compatible w.r.t the valuation (\eqref{4} is satisfied), if and only if it satisfies properties (P2) and (P3) in \cref{thm:feasible}.
\end{lemma}
The full proof for \cref{lem:ic} is left in \cref{ic proof}.

\begin{lemma}\label{lem:ir}
    For an incentive compatible contract, if (P5) in \cref{thm:feasible} is satisfied, then this contract must be individual rational. 
\end{lemma}
The full proof for \cref{lem:ir} is left in \cref{ir proof}.

Combining \cref{lem:ic}, \cref{lem:ir} and \cref{lem:IC2}, directly leads to \cref{thm:feasible}.

\subsection{Optimal Contract Design}

The goal of this paper is to design a feasible contract that maximizes the provider's utility. We solve this problem in two steps. In the first step, we first propose the optimal payment rule $p^*(x)=(p_k^*(x))$ for any $x=(x_{v^1},\cdots,x_{v^K})$, which is a function of a given allocation $x$.
In the second step, we substitute the payment rule $p^*(x)$ into programming \eqref{program}, resulting in a new programming whose variables are $x$.
By solving this programming, we ultimately obtained the optimal contract, which includes the optimal allocation rule $x^*$,  and, consequently, the optimal payment rule $p^*(x^*)$.

Based on this approach, we first consider a simpler scenario in which all clients have the same capacity $c$, and prove program \eqref{program} is equivalent to a linear programming problem. However, for the more complicated scenario where there are multiple capacity values, non-convex complementary constraints are additionally needed to ensure the property of resource greedyness. To handle these nonlinear constraints, we introduce slack variables to transform them into tractable smooth constraints, and then gradually approximate the original program.

\subsubsection{Scenario: $|C|=1$.}
In this case, all current clients have the same capacity $c$, and then we omit the superscript for convenience. The optimal payment rule is proposed in \cref{optimal payment-1}.
\begin{proposition} \label{optimal payment-1}
Suppose all clients have the same capacity. 
Then under a feasible contract, if its allocation profile is $x=(x_{v^1},\cdots,x_{v^K})$, the corresponding optimal payment $p^*_k(x)$ is 
\begin{align}
    p_k^*(x) = 
\begin{cases} 
v^{K} x_{v^{K}} - \sum_{j=k}^{K-1} v^j (x_{v^{j+1}} - x_{v^j}), &  k \leq K \\
v^{K} x_{v^{K}}, &k = K
\end{cases}\label{16}
\end{align} 

\end{proposition}
\begin{proof}
Because the contract is feasible, we have $p_{v^{K}}\geq v^{K} x_{v^{K}}$, by individual rationality constraint. 
In addition, the squeeze inequality in (P3) of \cref{thm:feasible} implies $p_{v^k}\geq p_{v^{k+1}}-v^k(x_{v^{k+1}}-x_{v^k})$.  
Because the objective function in programming \eqref{program} is linear in $\{p_{v^k}\}$, we can obtain the optimal payment rule as \eqref{16} by backward induction.
\end{proof}

Next, to derive the optimal allocation rule, we substitute equation \eqref{16} into equation \eqref{EUP}, which yields
\begin{align}
    &\max_{(x,p)} \mathbb E_{\forall i, \gamma_i \sim \Delta_i}[U_P] \nonumber\\
    =&  \sum_{k=2}^K \sum_{i=1}^n \left[ \lambda_{i}^{k} (\alpha - v^{k}) - \sum_{j=1}^{k-1} \lambda_{i}^{j} (v^{k} - v^{k-1}) \right] x_{v^k}+\sum_{i=1}^n  \lambda_{i}^{1} (\alpha - v^{1})x_{v^1} \notag\\
    &+ M \cdot \min\left\{0, \left( \sum_{i=1}^n \sum_{k=1}^K \lambda_{i}^{k} x_{v^k}- D \right) \right\}. \label{15}\\
    \text{s.t.} \quad & c\geq x_{v^1}\geq \cdots \geq x_{v^k} \cdots \geq x_{v^K}\geq 0    \notag 
\end{align}

To transform this optimization problem into a linear programming problem \eqref{15}, we need to linearize the nonlinear term involving the $min$ function in the objective function. Here's the step-by-step reformulation:

\noindent{\bf Step 1:} Introduce an Auxiliary Variable. Define a new variable $t$ to replace the nonlinear term:
\begin{eqnarray*}
    t=\min\left\{0, \left( \sum_{i=1}^n \sum_{k=1}^K \lambda_{i}^{k} x_{v^k}- D \right) \right\}.
\end{eqnarray*}

\noindent{\bf Step 2:} Linearize the $min$ Function. The $min$ function can be expressed with two linear inequalities:
\begin{eqnarray*}
    t\leq& 0; ~\mbox{and},~
    t\leq& \sum_{i=1}^n \sum_{k=1}^K \lambda_{i}^{k} x_{v^k}- D.
\end{eqnarray*}

Thus, the reformulated linear program is: 
\begin{align}
    &\max_{(x,p)} \mathbb E_{\forall i, \gamma_i \sim \Delta_i}[U_P] \label{reformal}\\ 
    \text{s.t.} \quad 
    &c\geq x_{v^1}\geq \cdots \geq x_{v^k} \cdots \geq x_{v^K}\geq 0\nonumber\\
    & t \leq \sum_{k=1}^K \sum_{i=1}^n \lambda_i^k x_{v^k} - D \nonumber\\
    & t\leq 0 \nonumber
\end{align}

By solving the above program \eqref{reformal}, the optimal allocations of $\{x^*_{v^k}\}$ are achieved.

\subsubsection{Scenario: $|C|>1$.} This case is significantly more complex than the one where $|C|=1$, because the conditions of resource greedyness can not be expressed as linear constraints.
Without loss of generality, assume $C=\left\{c^1, \cdots, c^L\right\}$. Similar to \cref{optimal payment-1}, we can construct an optimal payment rule under the case where $|C|>1$. 
\begin{proposition}
\label{optimal payment-2}
Under a feasible contract, if its allocation profile is $x =(x^{c^l}_{v^k})_{c^l\in C,v^k\in V}$, then the corresponding optimal payment $p^{c^l*}_{v^k}(x)$ is
\begin{align}
  p_{v_k}^{c^l*} = 
\begin{cases} 
v^{K} x_{v^{K}}^{c^l} - \sum_{j=k}^{{K}-1} v^j (x_{v^{j+1}}^{c^l} - x_{v^j}^{c^l}), & \forall k \leq {K} \\
v^{K} x_{v^{K}}^{c^l}, &  k = {K}
\end{cases}, \quad \forall c_l \in C.  \label{17} 
\end{align}
\end{proposition} 
The full proof for \cref{optimal payment-2} is left in \cref{payment proof}. 
Similarly to the case of $|C|=1$, by substituting \eqref{17} into \eqref{EUP}, and replacing the penalty term as $t=\min\left\{0, \left( \sum_{i=1}^n \sum_{k=1}^K\sum_{l=1}^L \lambda_{i}^{l,k} x^{c^l}_{v^k}- D \right) \right\}$, we have the below program.

\begin{align}
    \max_{(x,p)} &\mathbb E_{\forall i, \gamma_i \sim \Delta_i}[U_P] \label{P1}\\
    =\sum_{l=1}^L &\sum_{k=2}^K \sum_{i=1}^n \left[ \lambda_{i}^{l,k} (\alpha - v^{k}) - \sum_{j=1}^{k-1} \lambda_{i}^{l,j} (v^{k} - v^{k-1}) \right] x^{c^l}_{v^k}+\sum_{l=1}^L\sum_{i=1}^n  \lambda_{i}^{l,1} (\alpha - v^{1})x^{c^l}_{v^1}  +Mt\nonumber\\
    \text{s.t.} \quad & c^l\geq x^{c^l}_{v^1}\geq \cdots \geq x^{c^l}_{v^k} \cdots \geq x^{c^l}_{v^K}\geq 0  \quad \forall c^l\in C  \label{monotic}  \\
    &  x^{c^L}_{v^k}\geq \cdots \geq x^{c^l}_{v^k} \cdots \geq x^{c^1}_{v^k}\geq 0  \quad \forall v^k\in V  \label{monotic2} \\
    &(x^{c^{l'}}_{v^k}-x^{c^{l}}_{v^k})(x^{c^{l}}_{v^k}-c^{l})=0 \quad \forall v^k\in V, \quad \forall c^l,c^{l'}\in C \label{definitionredcahnge} \\
     & t \leq \sum_{l=1}^L\sum_{k=1}^K \sum_{i=1}^n \lambda_i^k x^{c^l}_{v^k} - D \label{t1} \\
    & t\leq 0 \label{t2}
\end{align}
in which \eqref{t1} and \eqref{t2} are the constraints for the penalty term, \eqref{monotic} and \eqref{monotic2} ensure the monotonic allocation, and \eqref{definitionredcahnge} ensure the resources greedy in \cref{def:RG}.
Clearly, the complementarity constraint of \eqref{definitionredcahnge} brings an obstacle to solve program \eqref{P1} 

Therefore,  we adopt the method proposed by Fletcher and Leyffer \cite{fletcherSolvingMathematicalPrograms2004a} to transform the complementarity constraints into nonlinear inequality constraints.

Firstly, \eqref{definitionredcahnge} is equivalent to  $(x^{c^{l+1}}_{v^k}-x^{c^{l}}_{v^k})(c^{l}-x^{c^{l}}_{v^k})\leq 0$ and $x^{c^{l+1}}_{v^k}-x^{c^{l}}_{v^k}\geq 0$, $c^{l}-x^{c^{l}}_{v^k}\geq 0$. The latter two have already been guaranteed in the linear constraints. 
This relaxation preserves the core characteristic of the complementarity constraint, which requires that at least one of the two factors is non-positive, while avoiding the restriction that the product must be strictly zero.

Next, we relax \eqref{definitionredcahnge} as $(x^{c^{l+1}}_{v^k}-x^{c^{l}}_{v^k})(c^{l}-x^{c^{l}}_{v^k})\leq \epsilon$, and then the program is transformed into \eqref{program-2}:
\begin{align} 
    &\max_{(x,p)} \mathbb E_{\forall i, \gamma_i \sim \Delta_i}[U_P] \label{program-2}\\
    \text{s.t.} \quad & c^l\geq x^{c^l}_{v^1}\geq \cdots \geq x^{c^l}_{v^k} \cdots \geq x^{c^l}_{v^K}\geq 0  \quad \forall c^l\in C \label{eq:order-valuation} \\
    & x^{c^L}_{v^k}\geq \cdots \geq x^{c^l}_{v^k} \cdots \geq x^{c^1}_{v^k}\geq 0  \quad \forall v^k\in V  \label{eq:order-capacity} \\
    &(x^{c^{l'}}_{v^k}-x^{c^{l}}_{v^k})(c^{l}-x^{c^{l}}_{v^k})\leq \epsilon \quad  \forall c^{l'} \ge c^{l},\quad c^{l'},c^l\in C \label{eq:epsilon} \\
     & t \leq \sum_{l=1}^L\sum_{k=1}^K \sum_{i=1}^n \lambda_i^k x^{c^l}_{v^k} - D \nonumber \\
    & t\leq 0 \nonumber
\end{align}

 We can solve this problem by invoking an existing NLP solver (such as filtermpec in \cite{fletcherSolvingMathematicalPrograms2004a} or knitro in \cite{byrd1999interior}).

Notice that when solving the programming problem, we relax \eqref{definitionredcahnge} to \eqref{eq:epsilon} by introducing a small positive tolerance parameter $\epsilon$. This relaxation may violate incentive compatibility (\cref{def:IC1}). 
We use \emph{regret} to measure how closely the solution of the relaxed programming problem approximates incentive compatibility, which is a common metric for characterizing the extent of approximation to equilibria~\cite{corleyRegretBasedAlgorithmComputing2020,erezRegretMinimizationConvergence2023}.
It represents the maximum excess utility a client can obtain by choosing a contract item that does not match their true type.
Specifically, the regret of the solution of the relaxed programming problem is defined as
\begin{equation*}
    \begin{aligned}
        \max_{\left(v^k,c^l\right)\in V\times C} \, \max_{\left(v^{k'},c^{l'}\right)\in \Omega\left(c^l\right)} \, \left(p_{v^{k'}}^{c^{l'}}-v^kx_{v^{k'}}^{c^{l'}}\right)-\left(p_{v^k}^{c^l}-v^kx_{v^k}^{c^l}\right), \\  
        \Omega\left(c^l\right) \deq\set*{(v^{k'},c^{l'})\in V\times C\mid x_{v^{k'}}^{c^{l'}}\le c^l}.
    \end{aligned}
\end{equation*}

\cref{lem:bounded-regret} shows that the regret is at most $O(\sqrt{\epsilon})$.

\begin{lemma}[Bounded Regret]\label{lem:bounded-regret}

The regret of the solution to the optimization problem \eqref{program-2} is at most $O(\sqrt{\epsilon})$.
Specifically, $p_{v^{k'}}^{c^{l'}}-v^kx_{v^{k'}}^{c^{l'}}-(p_{v^k}^{c^l}-v^kx_{v^k}^{c^l}) \le \sum_{k=1}^{K}v^k\cdot\sqrt{\epsilon}$ for all $v^k, v^{k'} \in V,c^l,c^{l'}\in C$ subject to $x_{v^{k'}}^{c^{l'}} \leq c^{l}$.
\end{lemma}

Due to the space limitation, the proof of \cref{lem:bounded-regret} is deferred to \cref{bounded proof}.

\section{Conclusion}
In the context of low utilization of leased GPU computing power resources, this study/research proposes a contract-based incentive mechanism to encourage current clients to return their idle resources. Considering that clients may strategically conceal their true idle computing capacity and its actual valuation, we design a contract framework that incentivizes truthful resource reporting. 
The contract design problem is formulated as an optimization problem aimed at maximizing the utility of resource providers while ensuring IC, IR, and clients’ maximum resource capacities. 
In this work, we assume finite discrete customer types, which is realistic and reasonable. In practice, resource providers can extract typical customer types from historical transaction data through cluster analysis. And in the process of computing resource rental, resource providers often provide users with limited types of packages, which is essentially a discrete type division. Such setting also effectively avoids the additional computational complexity brought about by continuous types. However, we still have to face the more complex IC and IR constraints caused by double information uncertainty.
Given the non-convex nature of the optimization problem, we first characterize the feasibility of the proposed contract and transform the constraint conditions accordingly. Due to the presence of two-dimensional decision variables, we derived the optimal strategies in an iterative manner, that is, we initially derive the optimal payment strategy, and we  then incorporate the payment strategy into the original problem and derive the optimal allocation strategy. Future work includes extending the current framework to consider a continuous client-type distribution rather than the discrete setting.

\newpage
\bibliographystyle{splncs04}

\newpage
\appendix

\section{Omitted Proof of \cref{lem:IC2}}\label{proof:lem:IC2}
\begin{proof}
The only if part is straightforward.
For the if part, it suffices to show that \eqref{4} and \eqref{5} implies \eqref{6}.
For any $v^k, v^{k'} \in V$ and $c^l, c^{l'} \in C$ with $x_{v^{k'}}^{c^{l'}} \leq c^l$, let us consider following two cases.

\noindent{{\bf Case 1.}} $x_{v^k}^{c^{l'}} \leq c^l$.
Let us suppose that client $i$ has her true type of $(v^k,c^l)$ and misreports a type $(v^k,c^{l'})$. Therefore, by \eqref{5}, we have $p_{v^k}^{c^l} - v^k x_{v^k}^{c^l} \geq p_{v^k}^{c^{l'}} - v^k x_{v^k}^{c^{l'}}$. 
Similarly, if the true type is $(v^k,c^{l'})$, while the false one is $(v^{k'},c^{l'})$, then \eqref{4} implies that $p_{v^k}^{c^{l'}} - v^k x_{v^k}^{c^{l'}} \geq p_{v^{k'}}^{c^{l'}} - v^k x_{v^{k'}}^{c^{l'}}$.
Combining these two inequalities, we have $p_{v^k}^{c^l} - v^k x_{v^k}^{c^l} \geq p_{v^{k'}}^{c^{l'}} - v^k x_{v^{k'}}^{c^{l'}}$, which is just \eqref{6}.

\noindent{{\bf Case 2.}} $x_{v^k}^{c^{l'}} > c^l$.
Under this case, we have $x_{v^k}^{c^{l'}}>c^l\geq x_{v^k}^{c^{l}}$, where the second inequality is from the resource feasibility condition. Thus, \cref{def:RG} ensures $x_{v^k}^{c^{l}}=c^l$ and $c^{l'}>c^l$. Combining the Monotonicity in Capacity, the condition of $x_{v^{k'}}^{c^{l'}} \leq c^l$, and resource feasibility, we have $c^l\geq x_{v^{k'}}^{c^{l'}}\geq x_{v^{k'}}^{c^{l}}$, due to \cref{def:RG}, we can also have $x_{v^{k'}}^{c^{l'}}= x_{v^{k'}}^{c^{l}}$. By \eqref{5}, we have $p_{v^{k'}}^{c^{l}} - v^{k'} x_{v^{k'}}^{c^{l}} \geq p_{v^{k'}}^{c^{l'}} - v^{k'} x_{v^{k'}}^{c^{l'}}$ and $p_{v^{k'}}^{c^{l'}} - v^{k'} x_{v^{k'}}^{c^{l'}} \geq p_{v^{k'}}^{c^{l}} - v^{k'} x_{v^{k'}}^{c^{l}}$. It follows that $p_{v^{k'}}^{c^{l'}} - v^{k'} x_{v^{k'}}^{c^{l'}}= p_{v^{k'}}^{c^{l}} - v^{k'} x_{v^{k'}}^{c^{l}}$, and thus $p_{v^{k'}}^{c^{l'}}=p_{v^{k'}}^{c^{l}}$. Combining \eqref{4} and two equalities of $p_{v^{k'}}^{c^{l'}}=p_{v^{k'}}^{c^{l}}$ and $x_{v^{k'}}^{c^{l'}}=x_{v^{k'}}^{c^{l}}$, we have
$p_{v^{k}}^{c^{l}} - v^{k} x_{v^{k}}^{c^{l}} \geq p_{v^{k'}}^{c^{l}} - v^{k} x_{v^{k'}}^{c^{l}}\geq p_{v^{k'}}^{c^{l'}} - v^{k} x_{v^{k'}}^{c^{l'}}$, which is just \eqref{6}.
\end{proof}

\section{Omitted Proof of \cref{lem:ic}}\label{ic proof}
\begin{proof}
By incentive compatibility w.r.t. the valuation, we have $p_{v^p}^{c^l} - v^p x_{v^p}^{c^l} \geq p_{v^q}^{c^l} - v^p x_{v^q}^{c^l}$ and 
$p_{v^q}^{c^l} - v^q x_{v^q}^{c^l} \geq p_{v^p}^{c^l} - v^q x_{v^p}^{c^l}.$
It follows $(v^p - v^q)(x_{v^p}^{c^l} - x_{v^q}^{c^l}) \leq 0$, by adding 
these two inequalities.
Thus, (P2) is established.
In addition, by rearranging these two inequalities from incentive compatibility w.r.t. the valuation, we have 
$p_{v^p}^{c^l} - p_{v^q}^{c^l} \geq v^p (x_{v^p}^{c^l} - x_{v^q}^{c^l})$ and
$p_{v^q}^{c^l} - p_{v^p}^{c^l} \geq v^q (x_{v^q}^{c^l} - x_{v^p}^{c^l})$, which is equivalent to the squeeze inequality in (P3).
\end{proof}

\section{Omitted Proof of \cref{lem:ir}}\label{ir proof}
\begin{proof}
    Suppose two clients have their true types $(v^p,c^l)$ and $(v^q,c^l)$ with $v^p
    \leq v^q$. Denote their utilities as $U_p= p_{v^p}^{c^l} - v^p x_{v^p}^{c^l}$ and $U_q = p_{v^q}^{c^l} - v^q x_{v^q}^{c^l}$, respectively. Because $v^p \leq v^q$ and due to incentive compatibility, we have $ U_p = p_{v^p}^{c^l} - v^p x_{v^p}^{c^l} \geq p_{v^q}^{c^l} - v^p x_{v^q}^{c^l} \geq p_{v^q}^{c^l} - v^q x_{v^q}^{c^l} = U_q$. 
    Therefore, the utility of client decreases as the valuation increases, implying $U_{v^K}=\min_{v^k\in V}\{U_{v^k}\}$. 
    The individual rationality can be obtained if $U_{v^K}\geq 0$.
\end{proof}

\section{Omitted Proof of \cref{optimal payment-2}}\label{payment proof}

\begin{proof}
Taking into account the capacity, for those with true capacity $c^l$, due to the IC constraint, we have $p_{v^{K}}^{c^l} - v^{K} x_{v^{K}}^{c^l}\geq0$, to minimize the payment, we have $p_{v^{K}}^{c^l} - v^{K} x_{v^{K}}^{c^l}=0$. Similar to the proof of proposition 1, we can derive that for any capacity $c^l$, we all have $p_{v^{K-1}}^{c^l} = p_{v^{K}}^{c^l} - v^{K-1} (x_{v^{K}}^{c^l} - x_{v^{K-1}}^{c^l})$, and by iterating, we can derive the payment function.

Then, we consider the payment for those with true capacity ${c^{l'}}$, according to \eqref{5}, for private valuation  $v^{K}$, we have $p_{v^{K}}^{c^{l'}} - v^{K} x_{v^{K}}^{c^{l'}} \geq p_{v^{K}}^{c^l} - v^{K} x_{v^{K}}^{c^l}$. To minimize the payment, we make $p_{v^{K}}^{c^{l'}} - v^{K} x_{v^{K}}^{c^{l'}} = p_{v^{K}}^{c^l} - v^{K} x_{v^{K}}^{c^l}$, thus $p_{v^{K}}^{c^{l'}}=v^{K} x_{v^{K}}^{c^{l'}}+p_{v^{K}}^{c^l} - v^{K} x_{v^{K}}^{c^l}=v^{K} x_{v^{K}}^{c^{l'}}$. Similar to the above deduction, we can also get the payment for other types with the same capacity $c^{l'}$, and by iterating, we can derive the payment.
\end{proof}

\section{Omitted Proof of \cref{lem:bounded-regret}}\label{bounded proof}

\begin{proof}
    The relaxed constraint violates the \eqref{rg.2} in \cref{def:RG}. By the optimal payment, we can easily deduce \eqref{4}.
    For a client with true type $(v^k,c^l)$, the difference of utility by selecting the contract item designed for type $(v^{k'},c^{l'})$ such that $x^{c^{l'}}_{v^{k'}}\leq c^l$ is: 
    \begin{equation*}
        \begin{aligned}
            p^{c^{l'}}_{v^{k'}}-v^k x^{c^{l'}}_{v^{k'}}-(p^{c^{l}}_{v^{k}}-v^k x^{c^{l}}_{v^{k}})
        \end{aligned}
    \end{equation*}
    
    According to \eqref{4} we can imply that 
    \begin{equation*}
        \begin{aligned}
            p^{c^{l'}}_{v^{k'}}-v^k x^{c^{l'}}_{v^{k'}}-(p^{c^{l}}_{v^{k}}-v^k x^{c^{l}}_{v^{k}})
            \leq p^{c^{l'}}_{v^{k'}}-v^k x^{c^{l'}}_{v^{k'}}-(p^{c^{l}}_{v^{k'}}-v^k x^{c^{l}}_{v^{k'}})
        \end{aligned}
    \end{equation*}

    According to the payment rule, the difference between $p^{c^{l'}}_{v^{k'}}$ and $p^{c^{l}}_{v^{k'}}$ is:
    \begin{align}
             &p^{c^{l'}}_{v^{k'}}-p^{c^{l}}_{v^{k'}}\nonumber\\
             =&v^K(x_{v^K}^{c^{l'}}-x_{v^K}^{c^{l}})-\sum_{j=k'}^{K-1}v^j(x_{v^{j+1}}^{c^{l'}}-x_{v^{j+1}}^{c^{l}})+\sum_{j=k'}^{K-1}v^j(x_{v^{j}}^{c^{l'}}-x_{v^{j}}^{c^{l}}) \nonumber
    \end{align}

    And \eqref{eq:epsilon} implies that for all $c^{l'}, c^l \in C$ and $v^k \in V$ such that $c^{l'}\geq c^l$, we have $x_{v^{k}}^{c^l}\ge c^l-\frac{\epsilon}{\left(x_{v^{k}}^{c^{l'}}-x_{v^{k}}^{c^l}\right)}$.
    Since $x_{v^{k}}^{c^l}\le x_{v^{k}}^{c^{l'}}\le c^l$, we have that $x_{v^{k}}^{c^{l'}}-x_{v^{k}}^{c^l}\le \frac{\epsilon}{\left(x_{v^{k}}^{c^{l'}}-x_{v^{k}}^{c^l}\right)}$.
    Thus $0\leq x_{v^{k}}^{c^{l'}}-x_{v^{k}}^{c^l}\le \sqrt{\epsilon}$.

    \noindent\textbf{Case 1:} $c^{l'}\geq c^l$. Combing \eqref{eq:order-capacity}, there is:
    \begin{align}
        &p^{c^{l'}}_{v^{k'}}-p^{c^{l}}_{v^{k'}}\nonumber\\
        \leq &v^K \sqrt{\epsilon}-\sum_{j=k'}^{K-1}v^j\cdot 0+\sum_{j=k'}^{K-1}v^j\cdot \sqrt{\epsilon} \nonumber\\
        =&v^K(x_{v^K}^{c^{l'}}-x_{v^K}^{c^{l}})-\sum_{j=k'}^{K-1}v^j(x_{v^{j+1}}^{c^{l'}}-x_{v^{j+1}}^{c^{l}})+\sum_{j=k'}^{K-1}v^j(x_{v^{j}}^{c^{l'}}-x_{v^{j}}^{c^{l}}) \label{shizi} \nonumber\\
             \leq &v^K \sqrt{\epsilon}+\sum_{j=1}^{K-1}v^j\cdot \sqrt{\epsilon} \nonumber\\
             =&\sum_{k=1}^{K}v^k\cdot \sqrt{\epsilon}\nonumber 
    \end{align}

    Therefore, we have that 
    \begin{equation*}
        \begin{aligned}
            &p^{c^{l'}}_{v^{k'}}-v^k x^{c^{l'}}_{v^{k'}}-(p^{c^{l}}_{v^{k}}-v^k x^{c^{l}}_{v^{k}}) \\
            \leq& p^{c^{l'}}_{v^{k'}}-v^k x^{c^{l'}}_{v^{k'}}-(p^{c^{l}}_{v^{k'}}-v^k x^{c^{l}}_{v^{k'}})\\
            \leq &\sum_{k=1}^{K}v^k \sqrt{\epsilon}-v^k(x^{c^{l'}}_{v^{k'}}-x^{c^{l}}_{v^{k'}})\\
            \leq &\sum_{k=1}^{K}v^k \sqrt{\epsilon}
        \end{aligned}
    \end{equation*}

    \noindent\textbf{Case 2:} $c^{l'} < c^l$. 
    Combing \eqref{eq:order-capacity}, there is:
    \begin{align}
        &p^{c^{l'}}_{v^{k'}}-p^{c^{l}}_{v^{k'}}\nonumber\\
             =&v^K(x_{v^K}^{c^{l'}}-x_{v^K}^{c^{l}})-\sum_{j=k'}^{K-1}v^j(x_{v^{j+1}}^{c^{l'}}-x_{v^{j+1}}^{c^{l}})+\sum_{j=k'}^{K-1}v^j(x_{v^{j}}^{c^{l'}}-x_{v^{j}}^{c^{l}})  \nonumber\\
             \leq &v^K \cdot 0-\sum_{j=k'}^{K-1}v^j\cdot (-\sqrt{\epsilon})+\sum_{j=k'}^{K-1}v^j\cdot 0 \nonumber\\
             = &\sum_{k=1}^{K-1}v^k\cdot \sqrt{\epsilon} \nonumber
    \end{align}

    Therefore, we have that 
    \begin{equation*}
        \begin{aligned}
            &p^{c^{l'}}_{v^{k'}}-v^k x^{c^{l'}}_{v^{k'}}-(p^{c^{l}}_{v^{k}}-v^k x^{c^{l}}_{v^{k}}) \\
            \leq& p^{c^{l'}}_{v^{k'}}-v^k x^{c^{l'}}_{v^{k'}}-(p^{c^{l}}_{v^{k'}}-v^k x^{c^{l}}_{v^{k'}})\\
            \leq &\sum_{k=1}^{K-1}v^k \sqrt{\epsilon}-v^k(x^{c^{l'}}_{v^{k'}}-x^{c^{l}}_{v^{k'}})\\
            \leq &\sum_{k=1}^{K-1}v^k \sqrt{\epsilon}-v^K(-\sqrt{\epsilon})\\
            \leq &\sum_{k=1}^{K}v^k \sqrt{\epsilon}
        \end{aligned}
    \end{equation*}

\end{proof}

\end{document}